 \documentclass[letterpaper, 10 pt, conference]{ieeeconf}

\IEEEoverridecommandlockouts
\usepackage[latin1]{inputenc}
\usepackage[english]{babel}
\usepackage{amsmath,amssymb}
\usepackage{graphicx,graphicx,verbatim}
\usepackage{verbatim}
\usepackage{color}
\newtheorem{thm}{Theorem}
\newtheorem{cor}{Corollary}
\newtheorem{lem}{Lemma} 
\newtheorem{assumption}{Assumption}

\newtheorem{defi}{Definition}
\newtheorem{prop}{Proposition}

\newtheorem{exx}{Example}
\newtheorem{remm}{Remark}

\newenvironment{remark}{\begin{remm}\rm }{\hfill \hspace*{1pt} \hfill $\lrcorner$\end{remm}}

\newenvironment{theorem}{\begin{thm} \rm }{\hfill \hspace*{1pt} \hfill $\lrcorner$\end{thm}}

\newenvironment{lemma}{\begin{lem}\rm }{\hfill \hspace*{1pt} \hfill $\lrcorner$\end{lem}}
\newenvironment{example}{\begin{exx}\rm }{\hfill \hspace*{1pt} \hfill $\lrcorner$ \end{exx}}
\newenvironment{proofof}{{\em Proof of }}{\hfill \hspace*{1pt}
\hfill $\blacksquare$}

\newcommand\real{\ensuremath{{\mathbb R}}}
\newcommand\realn{\ensuremath{{\mathbb{R}^n}}}

\newcommand\mymatrix[2]{\left[\begin{array}{#1} #2 \end{array}\right]}
\newcommand{\smallmat}[1]{\left[ \begin{smallmatrix}#1
    \end{smallmatrix} \right]}

\newcommand{\calA}{\mathcal{A}}
\newcommand{\calB}{\mathcal{B}}
\newcommand{\calC}{\mathcal{C}}

\newcommand{\calI}{\mathcal{I}}

\newcommand{\calS}{\mathcal{S}}

\newcommand{\calK}{\mathcal{K}}

\newcommand{\calY}{\mathcal{Y}}
\newcommand{\calX}{\mathcal{X}}

\begin{document}

\title{\LARGE \bf Differential positivity on compact sets}

\author{F. Forni
\thanks{
F. Forni is with the University of Cambridge, Department of Engineering, 
Trumpington Street, Cambridge CB2 1PZ, and with the Department 
of Electrical Engineering and Computer Science, 
University of Li{\`e}ge, 4000 Li{\`e}ge, Belgium, \texttt{ff286@cam.ac.uk}.
The research was supported by the Fund for Scientific Research FNRS
and by the Engineering and Physical Sciences Research Council under Grant EP/G066477/1.
The paper presents research results of the Belgian Network DYSCO
(Dynamical Systems, Control, and Optimization), funded by the
Interuniversity Attraction Poles Programme, initiated by the Belgian
State, Science Policy Office. The scientific responsibility rests with
its authors.} 
}

\date{\today}

\maketitle

\begin{abstract}  
The paper studies differentially positive systems, that is, systems whose linearization 
along an arbitrary trajectory is positive. Extending the results in \cite{Forni2014a_ver1}, 
we illustrate the use of differential positivity on compact forward invariant sets for the 
characterization of bistable and periodic behaviors. Geometric conditions for differential 
positivity are provided. The introduction of compact sets simplifies the use of differential 
positivity in applications.
\end{abstract}

\section{Introduction}

A linear system is positive
if some cone $\calK$ on the system state space is invariant
for the dynamics \cite{Bushell1973}.
Positivity strongly restricts the
behavior of a linear system. Under mild conditions,
the ray $\lambda v \in \calK$ given by 
the Perron-Frobenius eigenvector $v\in \calK$
is an attractor for the system 
dynamics, \cite{Birkhoff1957,Bushell1973}.
This fundamental property is exploited in a number of applications 
\cite{Farina2000,Moreau2004,Roszak2009,Rantzer2012}

Differential positivity extends linear positivity to the nonlinear setting.
A nonlinear system is differentially positive if
its linearization along trajectories
makes a cone (field) invariant \cite{Forni2014a_ver1}.
Differential positive systems are a large
class of systems encompassing 
monotone systems \cite{Angeli2003,Hirsch1995,Smith1995}.
Under mild conditions, the trajectories of a
differentially positive system converge
to a one dimensional attractor, 
a relevant property for the study of
bistable and periodic behaviors.
In comparison to linear positivity, this
attractor is not a ray, but a curve, possibly
closed in the presence of attractive limit cycles. 

In this paper we make 
differential positivity readily available
for the analysis of nonlinear systems
by deriving a number of geometric tools.
The problem of establishing the differential positivity
of a system is encoded into 
a set of pointwise geometric conditions to test.
Differential positivity is 
then used to derive 
novel methods for the analysis of 
simple attractors of nonlinear systems,
typically capturing bistable and periodic behaviors.
Compactness simplifies the use of differential positivity
in applications. In particular, 
the restriction to compact 
and forward invariant sets 
makes the geometric conditions for differential positivity
much simpler to verify in practice.

Section \ref{sec:cone_fields} introduces the
notion of cone fields and characterizes 
their representation as a set of inequalities.
The two large families of polyhedral and quadratic
cone fields are illustrated.
Section \ref{sec:differential_positivity} recalls 
the basic notions of differential positivity.
Section \ref{sec:geometric_conditions} 
provides a set of geometric conditions 
for testing the differential positivity of a system.
The use of the geometric conditions is illustrated on a cooperative
system and on a nonlinear pendulum.
The role of differential positivity 
for the analysis of 
the asymptotic nonlinear behavior 
is discussed in Section \ref{sec:asymptotic_behavior}
and illustrated on the Kuramoto model in 
Section \ref{sec:Kuramoto}. Conclusions follow.

\section{Cone fields}
\label{sec:cone_fields}

We recall some 
basic geometric notions on Riemannian manifolds
which will be useful to the discussion on
differential positivity.
Let $\mathcal{X}$ be a smooth $n$-dimensional manifold endowed with 
a Riemannian metric 
$\langle \cdot,\cdot \rangle_x:T_x\mathcal{X} \times T_x\mathcal{X} \to \mathbb{R}$
where $T_x\mathcal{X}$ denotes the tangent space at $x \in \mathcal{X}$. 
We will use $T\calX$ to denote the tangent bundle of $\calX$,
and 
$|\delta x|_x$ to denote $\sqrt{\langle \delta x, \delta x \rangle_x}$ 
for all $\delta x\in T_x \mathcal{X}$.
Furthermore, given any function $\psi:\calX\to\calY$
between manifolds,
$\partial \psi(x): T_x\mathcal{X} \to T_{\psi(x)}\mathcal{Y}$
will denote the differential of $\psi$ at $x\in \calX$.
Given any set $\calS \subseteq T_x\calX$ we will write
$\partial \psi(x) \calS := \{\psi(x)\delta x \,|\, \delta x \in \calS\} $.

To extend linear positivity into the
differential setting, we will exploit the
notion of conal manifold, that is, 
a manifold $\calX$ endowed with a cone field
\begin{equation}
 \calK_\varepsilon(x) \subseteq T_x\calX \qquad \forall x\in \calX.
\end{equation}
where $0 \leq \varepsilon \ll 1$.
At any $x$,  
$\calK_\varepsilon(x)$ is just a cone in the
vector space $T_x\calX$. The role of the parameter
$\varepsilon$ is clarified by the property
$
\calK_{\varepsilon_2}(x)\setminus\{0\} \subset \calK_{\varepsilon_1}(x)\setminus\{0\} 
$
if  
$
\varepsilon_1 < \varepsilon_2
$,
that we assume throughout the paper.
We use $\calK(x) := \calK_0(x)$.

For each $\varepsilon$, 
each cone $\calK_\varepsilon(x)$
is \emph{closed} and \emph{solid}, and satisfies 
(i) $\calK_\varepsilon(x)+\calK_\varepsilon(x) \subseteq \calK_\varepsilon(x)$, 
(ii) $\alpha \calK_\varepsilon(x) \subseteq \calK_\varepsilon(x) $ for all $\alpha > 0$, (iii) $\calK_\varepsilon(x) \cap -\calK_\varepsilon(x)=\{0\}$,
which make $\calK_\varepsilon(x)$ \emph{convex} and \emph{pointed}.
To avoid pathological cases, we assume
that for every $x_1,x_2\in \calX$, 
there exists a linear invertible mapping  
$\Gamma(x_1,x_2):T_{x_1}\calX \to T_{x_2}\calX$ 
such that, for any $\varepsilon$,
$\Gamma(x_1,x_2)\calK_\varepsilon(x_1) = \calK_\varepsilon(x_2)$.
Cone fields satisfying this property are said \emph{regular}.

To make use of a cone field in computations, the 
cone field $\calK_\varepsilon(x)$ will be represented by   
\begin{equation}
\label{eq:cone_fieldR}
\delta x \in \calK_\varepsilon(x)\setminus\{0\} \ \Leftrightarrow \
K_i\!\left(x,\frac{\delta x}{|\delta x|_x}\right) \geq \varepsilon   \,,\  \forall i\in \calI\subseteq\mathbb{N} 
\end{equation}
where each $K_i:T\calX \to \real$
is a \emph{smooth} function.
For simplicity, in what follows we will use
$K  \geq \varepsilon$ to denote 
the component-wise inequality
$K_i \geq \varepsilon$ for each $\forall i \in \calI $.
For the cone $\mathcal{K}(x)$, 
\eqref{eq:cone_fieldR} is equivalent $K_i(x,\delta x) \geq 0$.
The normalization $\frac{\delta x}{|\delta x|_x}$ 
in \eqref{eq:cone_fieldR} maps every ray $\lambda \delta x \in \calK(x)$, 
$\lambda > 0$, into the point $\frac{\delta x}{|\delta x|_x}$ of the 
unit sphere $\{v\in T_x\calX \,|\,|v|_x=1\}$. 
The normalization makes 
the representation independent
of the length of the tangent vectors $\delta x$. 
Examples of smoothly varying polyhedral 
and quadratic cone fields
are provided at the end of the section.
A simple illustration is in Figure \ref{fig:polyhedral_quadratic}.

A cone field carries naturally the 
useful notion of \emph{conal curve}
$\gamma: \real \to \calX$, which is an
integral curve of the cone field
\begin{equation}
\dot{\gamma}(s) \in \calK(\gamma(s)) 
\qquad \forall s \in \real \ .
\end{equation}
We make the standing assumption
that $|\dot{\gamma}(s)|_{\gamma(s)} = 1$ for any $s\in \real$.
From \eqref{eq:cone_fieldR}, a conal curve
satisfies $K(\gamma(s), \dot{\gamma}(s)) \geq 0$.

\begin{figure}[htbp]
\centering
\includegraphics[width=0.28\columnwidth]{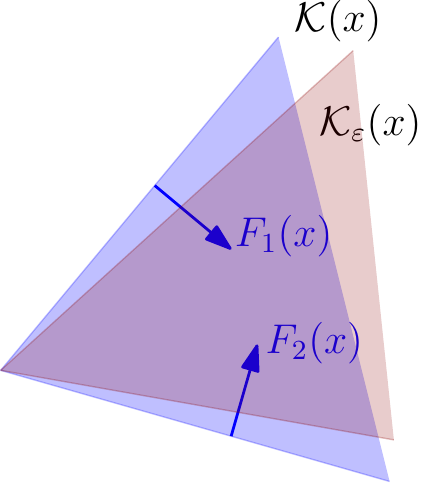}  
\hspace{3mm}
\includegraphics[width=0.47\columnwidth]{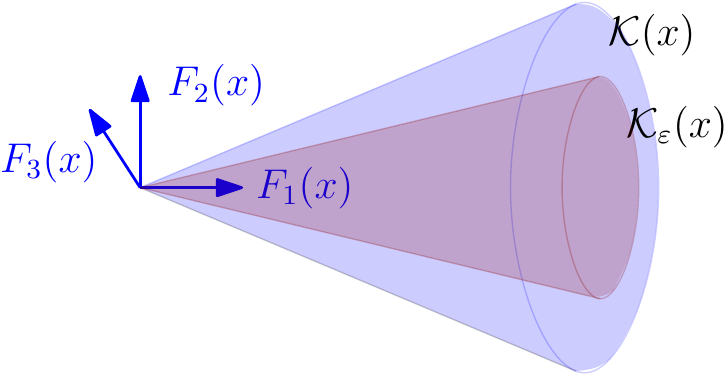}  
\caption{Polyhedral and quadratic cones at $T_x\calX$.}
\label{fig:polyhedral_quadratic}
\end{figure}

\begin{example}\emph{polyhedral cone fields.}
\label{example:polyhedral_cone}
Let $\calX$ be a smooth manifold of dimension $n$
endowed with a Riemannian metric $\langle \cdot,\cdot \rangle_x $.
For any $i \in \calI := \{1,\dots, m\}\subseteq\mathbb{N}$, $m\geq n$, 
define 
\begin{equation}
\label{eq:polyhedral_constraints}
K_i(x,\delta x) := \langle F_i(x), \delta x\rangle_x
\end{equation}
where $F_i(x)\in T_x\calX \setminus\{0\}$.
We assume that: 
\textbf{A1} for every pair of points $x_1, x_2\in \calX$, 
there exists a linear invertible isometry 
$T(x_1,x_2):T_{x_1}\calX \to T_{x_2}\calX$ such that 
$F_i(x_2) =  T(x_1,x_2) F_i(x_1)$;
\textbf{A2}
$\{F_1, \dots, F_m\}$
is a smooth full-rank distribution;
\textbf{A3} the set of constraints 
\eqref{eq:cone_fieldR},\eqref{eq:polyhedral_constraints}
is feasible for some $\overline{\varepsilon} > 0$.
Then, for all $0\leq \varepsilon <\overline{\varepsilon}$,
the cone field $\calK_\varepsilon(x)$ given by 
\eqref{eq:cone_fieldR},\eqref{eq:polyhedral_constraints}
is solid, pointed, convex and regular.

Assumption \textbf{A3} guarantees that, for
$\varepsilon < \overline{\varepsilon}$, 
$\calK_\varepsilon(x)$ is solid, pointed and convex by construction, 
since the distribution is full-rank.
For regularity, 
consider any pair of points $x_1,x_2\in \calX$, 
and take
$\langle F_i(x_1) , \frac{\delta x_1}{|\delta x_1|_{x_1}}\rangle_{x_1} = \varepsilon$
and $\delta x_2 = T(x_1,x_2) \delta x_1$. Then, 
$
\langle F_i(x_2), \delta x_2 \rangle_{x_2}
= 
\langle T(x_1,x_2) F_i(x_1), T(x_1,x_2)\delta x_1 \rangle_{x_2}
=
\langle F_i(x_1), \delta x_1 \rangle_{x_1}
= \varepsilon |\delta x_1|_{x_1}
$. Thus, 
$
\langle F_i(x_2) , \frac{\delta x_2}{|\delta x_2|_{x_2}}\rangle_{x_2} 
= 
\varepsilon \frac{|\delta x_1|_{x_1}}{|\delta x_2|_{x_2}} = \varepsilon
$
since $T(x_1,x_2)$ is an isometry.
\end{example}
\vspace{1mm}

\begin{example}\emph{quadratic cone fields.}
\label{example:quadratic_cone}
Let $\calX$ be a smooth manifold of dimension $n$
endowed with a Riemannian metric $\langle \cdot,\cdot \rangle_x $.
Consider $m\geq n$ vector fields $F_i(x)\in T_x\calX \setminus\{0\}$
and define 
\begin{equation}
\label{eq:quadratic_constraints_vectors}
\begin{array}{rcl}
K_1(x,\delta x) \!\!&\!\!:=\!\!&\!\!  \langle F_1(x), \delta x\rangle_x \vspace{2mm} \\
K_2(x,\delta x) \!\!&\!\!:=\!\!&\!\!  
\langle F_1(x), \!\delta x\rangle_{\!x}^{\!2} - 
\hspace{-5mm}
\sum\nolimits\limits_{i,j\in \{2,\dots,n\}}
\hspace{-5mm}
\langle F_i(x), \!\delta x \rangle_{\!x} \langle F_j(x), \!\delta x \rangle_{\!x}  \ .
\end{array}
\end{equation}

Assumptions \textbf{A1}-\textbf{A3}
with the additional condition \textbf{A4}
$
\langle F_1(x), F_i(x) \rangle_x = 0
$
for any $i > 1$, guarantee that 
the cone field 
\eqref{eq:cone_fieldR}, \eqref{eq:quadratic_constraints_vectors} is 
solid, pointed, convex and regular for $\varepsilon < \overline{\varepsilon}$.

Note that $K_2\geq 0$ characterizes a double cone which is 
refined to a pointed cone by $K_1\geq 0$. 
The cone fields $\calK$ and $\calK_\varepsilon$
are solid and convex by construction. Regularity 
follows from the observation that  
$\langle F_i(x_2), \delta x_2\rangle_{x_2}
= \langle F_i(x_1), \delta x_1\rangle_{x_1}$
for each $x_1,x_2 \in \calX$ and each 
$\delta x_2 = T(x_1,x_2) \delta x_1$.
\end{example}

\section{Differential positivity}
\label{sec:differential_positivity}

\subsection{Differential positivity in forward invariant regions}

A linear system $\dot{x} = A x$, $x\in \real^n$,
is positive if 
there exists a cone $\calK\subseteq \real^n$
which is forward invariant for the system dynamics,
i.e. $e^{At} \calK \subseteq \calK$ for $t\geq 0$,
\cite{Bushell1973}.
Differential positivity extends linear positivity
to nonlinear dynamics 
\begin{equation}
\label{eq:sys}
\Sigma \ : \ \dot{x} = f(x) \qquad x\in \calX, 
\end{equation}
by requiring that 
a given cone field is forward invariant for the 
prolonged dynamics \cite{Crouch1987},
\begin{equation}
\label{eq:prolonged_sys}
\delta \Sigma : \left\{
\begin{array}{rcl}
\dot{x} &=& f(x) \\
\dot{\delta x} &=& \partial f(x) \delta x  
\end{array}
\qquad (x,\delta x)\in T\calX \ .
\right.
\end{equation}

\eqref{eq:prolonged_sys} represents the linearization 
of $\Sigma$ along its trajectories. 
For simplicity we assume $f\!\in\! C^2$
and forward completeness of $\Sigma$. We use 
$\psi(t,x_0):\real\times\calX \to \calX$ to
denote the state reached at time $t$ by the trajectory of 
$\Sigma$ from the initial condition $x_0$. 
Indeed, $\psi(\cdot,x_0) \in \Sigma$. 
We also use $\psi_t(\cdot):\calX \to \calX$,
which maps any $x\in \calX$ into 
$\psi_t(x) := \psi (t,x)$.
Finally, for any $(x,\delta x) \in T\calX$, 
note that
$(\psi(\cdot,x), \partial_x \psi(\cdot,x) \delta x) \in \delta \Sigma$.

Revisiting the definitions in {\cite{Forni2014a_ver1},
consider any forward invariant region $\calC\subseteq \calX$. 
We say that $\Sigma$ is \emph{differentially positive} in 
$\calC$ with respect to the cone field $\calK(x)$
if for all $x\in \calC$ and $t\geq 0$, 
\begin{equation}
\partial \psi_t (x) \calK(x) \subseteq \calK(\psi_t(x)) \ .
\end{equation}
Differential positivity captures
the invariance of the cone field along the linearized dynamics. 
Furthermore, we say that $\Sigma$
is (uniformly) strictly differentially positive if 
it is differentially positive and
there exists $T>0$ and $\varepsilon > 0$ such that,
for all $x\in \calC$ and $t\geq T$,
\begin{equation}
\partial \psi_t (x) \calK(x) \subseteq \calK_\varepsilon(\psi_t(x)) \ .
\end{equation}
Strict differential positivity captures 
the contraction of the cone field 
$\calK(x)$ along the linearized dynamics,
as shown in Figure \ref{fig:diff_positivity}.

\begin{figure}[htbp]
\centering
\includegraphics[width=0.65\columnwidth]{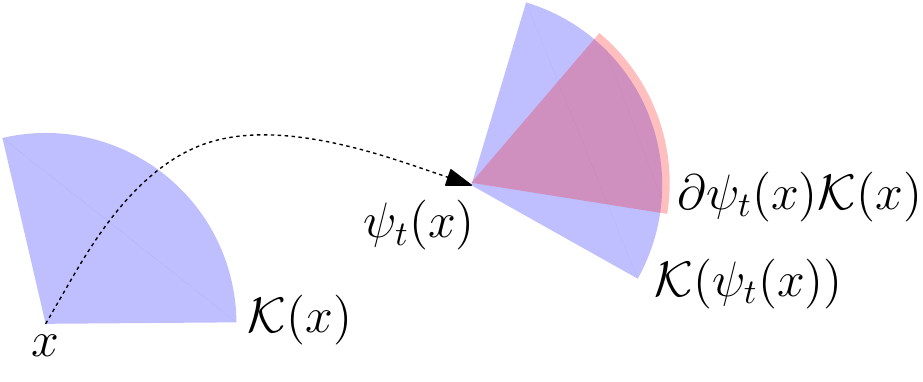}  
\caption{Differential positivity: forward
invariance of the cone field. 
Strict differential positivity: 
contraction of the rays of the cone field.}
\label{fig:diff_positivity}
\end{figure}

\subsection{Contraction of the Hilbert metric}

The contraction of the cone field along trajectories 
has a metric characterization based on the Hilbert metric, \cite{Bushell1973}. 
From \cite[Section VI]{Forni2014a_ver1},
for any given $x\in \calC$,
take any $\delta x, \delta y \in \calK(x)\setminus\{0\}$ and define 
$
M_{\calK(x)}(\delta x,\delta y) \!:=\!
 \inf \{\lambda \!\geq\! 0 \,|\, \lambda \delta y - \delta x \!\in\! \calK(x) \} 
$ \footnote{
$M_{\calK(x)}(\delta x,\delta y) := \infty$
when $ \{\lambda \in \real_{\geq 0} \,|\, \lambda \delta y - \delta x \in \calK(x) \} = \emptyset$.
}
and
$
m_{\calK(x)}(\delta x, \delta y) := 
 \sup \{\lambda \geq 0 \,|\, \delta x - \lambda \delta y  \in \calK(x) \} 
$.
The Hilbert metric $d_{\calK(x)}$ reads
\begin{equation}
\label{eq:lifted_hilbert_metric}
 d_{\calK(x)}( \delta x, \delta y) := \log \left(\frac{ M_{\calK(x)}(\delta x,\delta y)}{ m_{\calK(x)}(\delta x,\delta y)}\right) \ .
\end{equation}

$d_{\calK(x)}$ measures the distance between rays of the cone.
It defines a metric in 
$\calK(x) \cap \{\delta x \in T_x\calX\,|\, |\delta x|_x \!= \!1\} $.
Furthermore,
for any $\delta x,\delta y \in \calK(x)$,
$d_{\calK(x)}(\delta x,\delta y) = 0$
if and only if $\delta x=\lambda \delta y$ with $\lambda \geq 0$,
and $d_{\calK(x)}(\alpha \delta x,\beta \delta y)=
d_{\calK(x)}(\delta x,\delta y)$ for any $\alpha>0$ and $\beta>0$.

The contraction of the cone field along trajectories
is captured by the exponential convergence of the Hilbert metric,
as stated by the next lemma.
\begin{lemma} \cite[Theorem 2]{Forni2014a_ver1}.
\label{thm:dsch}
Let $\Sigma$ be a strictly differentially positive
system with respect to the cone field $\calK(x)$ in 
the forward invariant set $\calC\subseteq\calX$.
Then, there exist $k\geq1$ and $\lambda > 0$ such that, 
for all $x\in \calC$, $\delta x_1,\delta x_2 \in \calK(x)$, 
and $t\geq T$, 
\begin{equation}
\label{eq:Hilbert_contraction}
 d_{\calK(\psi_t(x))}
 (\partial \psi_t(x) \delta x_1,\partial \psi_t(x) \delta x_2) 
 \leq k e^{-\lambda(t-T)} \Delta 
\end{equation}
where $\Delta := 
\sup \{ d_{\calK(x)}(v_1,v_2) \,|\, v_1,v_2 \in \calK_\varepsilon(x) \} < \infty$.
\end{lemma}

Lemma \ref{thm:dsch} and the following 
(mild) technical assumption
are crucial  for the 
theorems of Section \ref{sec:asymptotic_behavior}.
\begin{assumption}
\label{assume:completeness}
$(\calK(x)\cap \{\delta x \in T_x\calX\,|\, |\delta x|_x =1\},d_{\calK(x)})$ is a \emph{complete metric space} for all $x\in \calC$.
\end{assumption}
The reader is referred to \cite[Section 4]{Bushell1973}, \cite[Section 2.5]{Lemmens2012},
or \cite{Zhai2011} for examples of complete metric spaces on cones.

\section{Geometric conditions}
\label{sec:geometric_conditions}

To provide geometric conditions for differential positivity
we first reformulate the property
using the representation \eqref{eq:cone_fieldR}. For instance, 
$\Sigma$ is differentially positive if 
for all $t \geq 0$ and all $(x,\delta x) \in T\calX$,
\begin{equation}
\label{eq:diff+}
K(x,\delta x)\geq 0 
\ \Rightarrow \ 
K(\psi_t(x),\partial \psi_t(x)\delta x) \geq 0 \ .
\end{equation}
In addition, $\Sigma$ is strictly 
differentially positive if 
there exists $T\!>\!0$ and $\varepsilon \!>\! 0$
such that, for all 
$t \!\geq\! T$ and all $(x,\delta x) \!\in\! T\calX$,
\begin{equation}
\label{eq:sdiff+}
K(x,\delta x)\geq 0 \ \Rightarrow \ 
K\!\left(\psi_t(x),\frac{\partial \psi_t(x)\delta x}{|\partial \psi_t(x)\delta x|_{\psi_t(x)}}\right) \geq \varepsilon \ .
\end{equation}
\eqref{eq:diff+} and \eqref{eq:sdiff+} 
capture the invariance and the contraction 
of the cone field along trajectories, 
which leads to the following pointwise geometric 
conditions for differential positivity.

\begin{theorem}
\label{thm:diff+_geom}
$\Sigma$ is differentially positive in the forward invariant 
set $\calC\subseteq\calX$ with respect to $\calK(x)$ if 
for any $(x,\delta x)\in T\calX$ such that $K(x,\delta x) \geq 0$
and $x\in \calC$,
\begin{equation}
\label{eq:diff+_geom}
K_i(x,\delta x) = 0 
\ \Rightarrow \ 
\partial K_i(x,\delta x) 
\mymatrix{c}{f(x) \\ \partial f(x) \delta x} \geq 0 \ .
\end{equation}
\end{theorem}
\begin{proof}
Whenever $(\psi(\cdot,x),\partial_x \psi(\cdot,x)\delta x)\in\delta\Sigma$
reaches the boundary of the cone field at time $t$,
we have 
$K(\psi(t,x),\partial_x \psi(t,x)\delta x) \geq 0$ and 
$K_i(\psi(t,x),\partial_x \psi(t,x)\delta x = 0$, 
for some $i \in \calI$. From \eqref{eq:diff+_geom},
$\frac{d}{dt} 
K_i(\psi(t,x),\partial_x \psi(t,x)\delta x \geq 0$, 
thus 
$K_i(\psi(t,x),\partial_x \psi(t,x)\delta x)$
either grows positive or remains at zero.
\end{proof}

For strict differential positivity we need to
take into account vectors on the unit sphere.
For instance, for any $\delta x \in T_x\calX$, 
consider $\vartheta := \frac{\delta x}{|\delta x|_x}$.
From \eqref{eq:prolonged_sys}, 
\begin{equation}
\label{eq:normalized_sys}
\dot{\vartheta} = \left(\partial f(x) - \lambda(x,\vartheta)\right) \vartheta
\end{equation}
where  $\lambda(x,\vartheta)$ normalizes the
action of the operator $\partial f(x)$ to guarantee that 
any trajectory $\vartheta(\cdot)$ of
\eqref{eq:normalized_sys} from $|\vartheta(0)|_{x(0)}=1$
satisfies $|\vartheta(t)|_{x(t)}=1$ for all $t> 0$.
For example, given the representation
$|\delta x|_x := (\delta x^T G(x) \delta x)^{\frac{1}{2}}$, where
$G(x)$ is the Riemannian tensor in local coordinates, 
$\lambda(x,\vartheta) := 
\frac{1}{2} \vartheta^T \left( G(x) \partial f(x) + \partial f(x)^T G(x) + 
\sum_{i=1}^n \partial_{x_i} G(x) f(x)_i \right) \vartheta$.
($f(x)_i$ is the $i$th component of the vector $f(x)$).

\begin{theorem}
\label{thm:sdiff+_geom}
$\Sigma$ is strictly differentially positive 
in the forward invariant set $\calC\subseteq\calX$ 
with respect to $\calK(x)$ if 
there exist $T>0$ and $\varepsilon>0$ such that,
for any $(x,\vartheta)\in T \calX$ that satisfies
$K(x,\vartheta) \geq 0$, $x\in \calC$ and $|\vartheta|_x = 1$,
\begin{equation}
\label{eq:sdiff+_geom}
\begin{array}{l}
 0 \leq K_i\left(x,\vartheta \right)  \leq \varepsilon 
\ \ \Rightarrow  \vspace{1mm} \\
 \partial K_i\!\left(x,\vartheta \right) \!
\mymatrix{c}{\!f(x)\! \\ \!  \left(\partial f(x) \!-\! \lambda(x,\vartheta)\right) \vartheta \!} \geq \dfrac{\varepsilon}{ T} \ .
\end{array} \vspace{-3mm}
\end{equation} \vspace{2mm}
\end{theorem}
\begin{proof}
Consider any trajectory 
$(\psi(\cdot,x),\partial_x \psi(\cdot,x)\delta x)\in\delta\Sigma$.
For simplicity, define 
$\vartheta(t) := \frac{\partial_x \psi(t,x)\delta x}{|\partial_x \psi(t,x)\delta x|_{\psi(t,x)}}$ and
$x(t) := \psi(t,x)$.
Differential positivity follows from the argument 
of Theorem \ref{thm:diff+_geom}. 
For strict differential positivity, 
suppose that 
for all $t \in [0,T]$,
$0 \leq K_i(x(t),\vartheta(t)) < \varepsilon$. Then, \eqref{eq:sdiff+_geom}
guarantees that
$K_i(x(T),\vartheta(T)) 
= K_i(x(0),\vartheta(0)) + \int_0^T \frac{d}{dt} K_i(x(t),\vartheta(t))dt  
\geq \int_0^T \frac{\varepsilon}{ T} dt 
= \varepsilon$. A contradiction.
\end{proof}
\vspace{1mm}

Because of the normalization, 
verifying \eqref{eq:sdiff+_geom} 
may be a daunting task in practice. However, 
those conditions become simpler when the
forward invariant region $\calC$ is \emph{compact}.

\begin{theorem}
\label{thm:sdiff+_geom_simple}
$\Sigma$ is strictly differentially positive in 
the compact and forward invariant set $\calC\subseteq\calX$
with respect to $\calK(x)$ if 
for
any $(x,\vartheta)$ such that $K(x,\vartheta) \geq 0$, $x\in \calC$ and $|\vartheta|_x = 1$,
\begin{equation}
\label{eq:sdiff+_geom_simple}
K_i\left(x,\vartheta \right)  = 0
\ \ \Rightarrow \ \
\partial K_i\!\left(x,\vartheta \right) \!
\mymatrix{c}{\! f(x)\! \\ \!  \partial f(x) \vartheta \!} > 0 \ .
\vspace{-3mm}
\end{equation} \vspace{2mm}
\end{theorem}
\begin{proof}
\textbf{(i)}~For $K(x,\vartheta) \geq 0$, consider
$K_i\left(x,\vartheta \right)  = 0$. Then, $K_i\left(x,\rho \vartheta \right)  = 0$ for any $\rho > 0$,
thus $[\partial_\vartheta K_i\left(x,\vartheta \right)] \vartheta = 0$. Since
$\calC$ is a compact set, by continuity, there exists $k_1>0$ (sufficiently large) 
and $\varepsilon>0$ (sufficiently small)  
such that
\begin{equation}
\label{eq:sdiff+_geom_simple_boundary}
0 \leq K_i\left(x,\vartheta \right)  \leq \varepsilon
\ \Rightarrow \ [\partial_\vartheta K_i\left(x, \vartheta \right)] \vartheta \leq k_1 \varepsilon \ .
\end{equation}

\textbf{(ii)}~Exploiting the compactness of $\calC$ again,
\eqref{eq:sdiff+_geom_simple} guarantees that 
there exist $k_2 > 0$ and a small $\varepsilon > 0$ such that 
$
K_i\left(x,\vartheta \right)  = 0
\Rightarrow  
\partial K_i\!\left(x,\vartheta \right) \!
\mymatrix{cc}{\! f(x)^T\! & \!  \partial f(x)^T \vartheta \!}^T \geq 
k_2 + \overline{\lambda}k_1 \varepsilon 
$
where
$\overline{\lambda} := \max_{x\in \calC,|\vartheta|_x =1, \calK(x,\vartheta) \geq 0}  \lambda(x,\vartheta)$
and $\lambda(x,\vartheta)$ refers to \eqref{eq:normalized_sys}.
By continuity, for $\varepsilon$ sufficiently small
there exists $k_3>0$ such that
\begin{equation}
\label{eq:sdiff+_geom_simple3}
\begin{array}{l}
0 \leq K_i\left(x,\vartheta \right) \leq \varepsilon
\ \ \Rightarrow  \vspace{1mm} \\
\partial K_i\!\left(x,\vartheta \right) \!
\mymatrix{c}{\! f(x)\! \\ \!  \partial f(x) \vartheta \!} \geq 
k_2+ \overline{\lambda}k_1 \varepsilon - k_3 \varepsilon \geq \frac{k_2}{2}+ \overline{\lambda}k_1\ .
\end{array}
\end{equation}
for any $x\in \calC$ and 
$|\vartheta|_x = 1$ such that $\calK(x,\vartheta) \geq 0$.

\textbf{(iii)}~We follow now the 
proof of Theorem \ref{thm:sdiff+_geom},
combining \textbf{(i)} and \textbf{(ii)}.
Consider any trajectory 
$(\psi(\cdot,x),\partial_x \psi(\cdot,x)\delta x)\in\delta\Sigma$,
define 
$\vartheta(t) := \frac{\partial_x \psi(t,x)\delta x}{|\partial_x \psi(t,x)\delta x|_{\psi(t,x)}}$ and
$x(t) := \psi(t,x)$.
Take $T := \frac{2 \varepsilon}{k_2}$ 
and suppose that for all $t \in [0,T]$,
$0 \leq K_i(x(t),\vartheta(t)) < \varepsilon$.
Combining \eqref{eq:normalized_sys}, \eqref{eq:sdiff+_geom_simple_boundary} 
and \eqref{eq:sdiff+_geom_simple3}, we get
\begin{equation}
\begin{array}{l}
\frac{d}{dt} K_i(x(t),\vartheta(t)) \ \geq \vspace{1mm}\\
\geq \frac{\varepsilon }{T}+ \overline{\lambda}k_1 \varepsilon 
- \underbrace{\lambda(x(t), \vartheta(t))}_{\leq \overline{\lambda}}
\underbrace{[\partial_{\vartheta(t)} K_i(x(t),\vartheta(t))] \vartheta(t)}_{\leq k_1\varepsilon} 
\geq \frac{\varepsilon}{T} \vspace{-3mm}
\end{array} \vspace{2mm}
\end{equation}
which leads to a contradiction by integration over the interval $[0,T]$,
as in the proof of Theorem \ref{thm:sdiff+_geom}.
\end{proof}

Theorem \ref{thm:sdiff+_geom_simple} is illustrated in the following
examples.
\begin{example}\emph{Cooperative systems.}
\label{example:cooperativity}
A nonlinear system $\Sigma $ given by $\dot{x} = f(x)$, $x\in \realn$,
is cooperative if its Jacobian $\partial f(x)$ has
nonnegative off-diagonal elements, \cite{Smith1995}.
Cooperative systems are differentially positive systems
in $\calX := \realn$ with respect to the constant polyhedral 
cone field given by the positive orthant, that is, 
$\calK(x) := \real^n_+$ for any $x\in \calX$
(Theorem \ref{thm:diff+_geom}).
Given any compact and forward invariant 
region $\calC\subseteq\realn$, $\Sigma$
is strictly differentially positive if
the off-diagonal elements of $\partial f(x)$ 
are strictly positive
(Theorem \ref{thm:sdiff+_geom_simple}). To see this, 
consider the standard inner product in $\realn$
and the cone field \eqref{eq:polyhedral_constraints}
given by $F_i(x) := e_i$ for each $x\in \calC$
and $i \in \calI := \{1,\dots,n\}$,
where $e_i$ is the canonical base.
Then, $\delta x \in \calK(x)$ reads
$e_i^T \delta x \geq 0$ for all $i$. 
\eqref{eq:sdiff+_geom_simple} reads
$e_i^T \delta x = 0$ $\Rightarrow$ $e_i^T \partial f(x) \delta x > 0
$
for $|\delta x|_2 = 1$, 
which is equivalent to
$
e_i^T \partial f(x) e_j > 0
$
for all
$
i,j \in \calI$, $i\neq j$.
\end{example}

\begin{example}\emph{Differential positivity of the pendulum.}
Consider the nonlinear pendulum given by the equations
$\dot{\vartheta} = v$, 
$\dot{v} = -\sin(\vartheta) - k v + u$.
The linearization reads 
$\dot{\delta \vartheta} = \delta v$, 
$\dot{\delta v} = -\cos(\vartheta)\delta \vartheta  - k \delta v $.
Theorem \ref{thm:sdiff+_geom_simple} guarantees that
for any $k>2$, any input $u$, and any compact and forward invariant region $\calC$, 
the pendulum is strictly differentially positive in $\calC$
with respect to the cone field \eqref{eq:polyhedral_constraints} given by
\begin{equation}
 \delta \vartheta \geq 0 \qquad \delta \vartheta + \delta v \geq 0
\end{equation}
For instance, for $|\smallmat{\delta \vartheta & \delta v}^T| = 1$,
\textbf{(i)} $\delta \vartheta = 0, \delta \vartheta + \delta v > 0$.
Then $\dot{\delta \vartheta} = \delta v > - \delta \vartheta > 0$.
\textbf{(ii)} $\delta \vartheta > 0, \delta \vartheta + \delta v = 0$.
Then, $\dot{\delta \vartheta} + \dot{\delta v} =
 \delta v - \cos(\vartheta) \delta \vartheta - k \delta v
 \geq (k-1-1) \delta \vartheta > 0$. 

Revisiting \cite[Section VIII]{Forni2014a_ver1},
Theorem \ref{thm:PB} below can be used to establish
the existence of limit cycles for $u > 1$.
\end{example}

\section{Asymptotic behavior}
\label{sec:asymptotic_behavior}

\subsection{Stable attractors}
It is well known that under mild conditions
almost every bounded trajectory
of a monotone system converges to a fixed point
\cite{Hirsch1988,Smith1995}.
This fundamental result 
has been recently revisited in \cite[Corollary 5]{Forni2014a_ver1},
with a new proof based on differential positivity.
The next theorem extends this result to a larger class 
of systems.

Given a compact set $\calC\subseteq\calX$ we say that a 
conal curve $\gamma : \real \to \calX$
\emph{intersects the boundary of $\calC$ twice} 
if for any $s\in \real$ such that $\gamma(s)\in \calC$, 
there exists an interval $\underline{s}\leq s \leq \overline{s}$ 
such that 
$\gamma(\underline{s}), \gamma(\overline{s}) \notin \calC$.
 
\begin{theorem}
\label{thm:bistability}
Under Assumption \ref{assume:completeness},
consider a strictly differentially positive system $\Sigma$ 
with respect to the cone field $\calK(x)$
in a compact forward invariant region $\calC\subseteq\calX$. Suppose that
every conal curve $\gamma:\real \to \calX$ 
intersects the boundary of $\calC$ twice.
Then, from almost every initial condition in $\calC$
the trajectories of $\Sigma$ converge asymptotically to a fixed point. 
\end{theorem}

The next theorem exploits
the combination of differential positivity with 
the existence of an invariant vector field
$v(x)\in \calK_\varepsilon(x)$. It shows
that the trajectories of the system converge
asymptotically to a one dimensional
attractor given by the image of 
an integral curve of $v(x)$.

\begin{theorem}
\label{thm:one-dimensional-attractor}
Under Assumption \ref{assume:completeness},
consider a strictly differentially positive system $\Sigma$ 
with respect to the cone field $\calK(x)$
in a compact forward invariant region $\calC\subseteq\calX$. 
Let $\varepsilon>0$ and suppose that 
there exists a complete vector field $v(x) \in \calK_\varepsilon(x)\setminus \{0\}$
such that 
\begin{equation}
\label{eq:boundedness}
	\limsup_{t\to\infty} |\partial \psi_t(x) v(x)|_{\psi_t(x)} < \infty \ ;
\end{equation}
\begin{equation}
\label{eq:directional_invariance}
	v(\psi_t(x)) = \frac{\partial \psi_t(x) v(x)}{|\partial \psi_t(x) v(x)|_{\psi_t(x)}} 
	\qquad \forall x\in \calC, \forall t\geq 0 \ .
\end{equation}
Then, there exists an integral curve of $v(x)$ whose
image is an attractor for all the trajectories
of $\Sigma$ from $\calC$.
\end{theorem}

Finding the vector field that satisfies 
\eqref{eq:boundedness} and \eqref{eq:directional_invariance} 
can be difficult in general. 
However, the presence of \emph{symmetries} in the system
makes \eqref{eq:boundedness} and \eqref{eq:directional_invariance} 
tractable conditions.
An example is given by  consensus dynamics 
\cite{Olfati-Saber2007,Sepulchre2010a}
where $n$ agents communicate with their neighborhoods 
to achieve consensus, typically given by 
the manifold $x_1=\dots=x_n$. 
The invariance of the consensus manifold 
dictates the symmetry $\partial f(x) \mathbf{1} = 0$,
$\mathbf{1} := \smallmat{1 & \dots & 1}^T \in T_x\calX $,
which makes 
\eqref{eq:boundedness} and \eqref{eq:directional_invariance} 
trivially verified by $v(x) := \frac{\mathbf{1}}{\sqrt{n}}$.
This observation is used in Section \ref{sec:Kuramoto}
for the analysis of the Kuramoto model.

Replacing $v(x)$ in Theorem \ref{thm:one-dimensional-attractor}
with the system vector field $f(x)$, the next theorem gives 
conditions for the 
existence of attractive limit cycles.
This result is compatible
with Theorem \ref{thm:one-dimensional-attractor}, since
conal curves can be closed curves.

\begin{theorem}
\label{thm:PB}
Under Assumption \ref{assume:completeness},
consider a strictly differentially positive system $\Sigma$ 
with respect to the cone field $\calK(x)$
in a compact forward invariant region $\calC\subseteq\calX$.
Suppose that $\calC$ does not contain any fixed point. 
If $f(x) \in \calK_\varepsilon(x)\setminus\{0\}$ 
for any $x\in \calC$, 
for some $\varepsilon>0$, then
there exists a unique attractive periodic orbit contained in $\calC$.
\end{theorem}
\begin{remark}
Converse results for hyperbolic limit cycles can be found in \cite{Mauroy2015}.
Theorem \ref{thm:PB} revisits \cite[Corollary 2]{Forni2014a_ver1},
which requires differential positivity in the whole manifold $\calX$,
a condition weakened by Theorem \ref{thm:PB}.
The key step for this result is a new proof that
does not use the so-called
Perron-Frobenius vector field of \cite[Section VI]{Forni2014a_ver1}.
\end{remark}

\subsection{Proofs}

\begin{proofof}\emph{Theorem \ref{thm:bistability}.}
Suppose that for some $x\in\calX$, the trajectory $\psi(\cdot,x)$
does not converge to a fixed point and denote by $\omega(x)$
the $\omega$-limit set of $x$.
Then, there exists a sequence of time instant $t_k \to \infty$
as $k\to \infty$ and $c>0$
such that $|f(\psi(t_k,x))| > c$, for all $k\in \mathbb{N}$.
By continuity, since $\psi(t,x) \in \calC$ for all $t\geq 0$, 
there exists a small constant $\rho >0$ such that 
$|f(\psi(t,x))| > c$ for all $t\in[t_k-\rho,t_k+\rho]$.

Recall also that 
$(\psi(\cdot,x),f(\psi(\cdot,x))\in\delta\Sigma$ since
$\frac{d}{dt} f(\psi(t,x)) = \partial f(\psi(t,x)) f(\psi(t,x))$. 
Thus, either \textbf{(i)} there exists $\tau\geq 0$ such that 
$f(\psi(t,x)) \in -\calK(\psi(t,x))\cup \calK(\psi(t,x))$ 
for all $t\geq \tau$, or \textbf{(ii)}
$f(\psi(t,x))\notin -\calK(\psi(t,x))\cup \calK(\psi(t,x))$ for all $t\geq 0$.

For \textbf{(i)}, note that the points $\psi_{t+\tau}(x)$ for $t\geq 0$ belong to 
the image of a conal curve, 
since $\frac{d}{dt}\psi_{t+\tau}(x) = f(\psi_{t+\tau}(x)) \in \calK(\psi_{t+\tau}(x))$ 
(or $-f(\psi_{t+\tau}(x)) \in \calK(\psi_{t+\tau}(x))$) for any $t\geq 0$.
Then, since $|f(\psi_{t+\tau}(x))| \geq c$ for all $t+\tau\in[t_k-\rho,t_k+\rho]$,
exploiting the fact that every conal curve intersects the boundary
of $\calC$ twice, there exists a time $T> 0$ such that $\psi_{T+\tau}(x)\notin \calC$.
This contradicts the forward invariance of $\calC$.

For \textbf{(ii)}, we show that the basin of attraction of $\calA:=\omega(x)$ 
has dimension $n-1$ at most. 
We need a preliminary result. Consider any $\delta x \in \calK_\varepsilon(x)$
for some $\varepsilon > 0$.
Then, for $\alpha>0$ sufficiently large
$f(x) + \alpha \delta x \in \calK(x)$.
Thus, by projective contraction,
$
\lim\nolimits\limits_{t\to\infty} d_{\calK(\psi_t(x))}(\partial \psi_t(x)[f(x) + \alpha \delta x], \partial \psi_t(x) \alpha \delta x)  
=
\lim\nolimits\limits_{t\to\infty} 
d_{\calK(\psi_t(x))}(\alpha \partial \psi_t(x)\delta x + f(\psi_t(x)),\alpha \partial \psi_t(x)  \delta x) 
= 0
$.
Since $ |f(\psi_{t_k}(x))| \geq c$, it follows that 
\begin{equation}
\label{eq:local_instability}
\lim_{t\to\infty} |\partial \psi_t(x)\delta x|_{\psi_t(x)} = \infty \qquad \forall \delta x \in \mbox{int} \calK(x) \ .
\end{equation}

By contradiction, suppose now that the basin of attraction
$\calB_\calA$ has dimension $n$. 
Consider a conal curve $\gamma$ such that $\gamma(0) \in \calA$ and 
$\gamma(s) \in \calB_\calA\setminus\{A\}$ for all $s\in [0,\overline{s}]\subseteq\real$.
By assumption,
for all $s\in[0,\overline{s}]$, $ \psi_t (\gamma(s))$ converges asymptotically to $\calA$ 
as $t\to \infty$. Moreover, $f(\psi_t (\gamma(s))) \notin \calK(\psi_t (\gamma(s)))$
for all $t\geq 0$,  
and $ |f(\psi_{t_k}(\gamma(s)))| \geq c$ for $k$ sufficiently large.
From \eqref{eq:local_instability}, it follows that
\begin{equation}
\label{eq:local_instability2}
\lim_{t\to\infty} |\partial \psi_t(x) \dot{\gamma}(s)|_{\psi_t(\gamma(s))} = \infty.
\end{equation}
From \eqref{eq:local_instability2},
the length of the curve $\psi_t (\gamma(\cdot))$ grows unbounded as $t \to \infty$,
since $\ell(\psi_t(\gamma(\cdot))) = \int_0^1 |\frac{d}{ds} \psi_t (\gamma(s)) ds|_{\psi_t(\gamma(s))} 
= \int_0^1 |\partial \psi_t (\gamma(s)) \dot{\gamma}(s) |_{\psi_t(\gamma(s))}  ds$.
It follows that there exists $T > 0$ such that 
$\psi_T(\gamma(\bar{s})) \notin \calC$, contradicting the forward invariance of the set.
\end{proofof}

\begin{proofof}\emph{Theorem \ref{thm:one-dimensional-attractor}.}
[\textbf{(i)} boundedness]
Under the assumptions of the theorem, we prove that
for any given $\delta x \in T_x\calX$,
\begin{equation}
\label{eq:local_stability}
\lim_{t\to\infty} |\partial \psi_t(x)\delta x|_{\psi_t(x)} < \infty.
\end{equation}
Consider first the case $\partial \psi_t(x) \delta x \notin -\calK(\psi_t(x))\cup \calK(\psi_t(x))$ for all $t \geq 0$.
Then for $\alpha >0 $ sufficiently large
$ \delta x + \alpha v(x) \in \calK(x)$, thus Lemma \ref{thm:dsch} 
gives
$
\lim\nolimits\limits_{t\to\infty}
d_{\calK(\psi_t(x))}(\partial \psi_t(x)[\alpha  v(x) +  \delta x], \partial \psi_t(x) v(x))  
 = 
\lim\nolimits\limits_{t\to\infty} 
d_{\calK(\psi_t(x))}(\alpha\partial \psi_t(x)  v(x)  \!+\! \partial \psi_t(x)\delta x ,\alpha\partial \psi_t(x)  v(x)) 
 = 0 
 $.
Since $|\partial \psi_t(x)  v(x)|_{\psi_t(x)}$ is bounded,
it follows that 
\begin{equation}
\label{eq:local_convergence}
\lim_{t\to\infty} |\partial \psi_t(x)\delta x|_{\psi_t(x)} = 0 \ .
\end{equation}

Consider now the case $\partial \psi_T(x) \delta x \in \calK(\psi_T(x))$ for some $T\geq 0$.
If $\partial \psi_T(x) \delta x \in -\calK(\psi_T(x))$, consider the symmetric trajectory
$\partial \psi_T(x) [-\delta x] \in \calK(\psi_T(x))$.
Suppose that $\lim\nolimits\limits_{t\to\infty} |\partial \psi_t(x)\delta x|_{\psi_t(x)} = \infty$.
Then, by projective contraction
$
d_{\calK(\psi_{t+T}(x))}( 
\partial \psi_{t+T}(x)\delta x , 
\partial \psi_{t+T}(x)  v(x) ) = 0 
$.
Therefore, for $t$ sufficiently large, 
$\partial \psi_{t+T}(x)v(x)$ is almost parallel to $\partial \psi_{t+T}(x)\delta x$. 
By linearity of $\partial \psi_{t+T}(x)$, it follows that 
$\partial \psi_{t+T}(x) v(x)$
grows unbounded, contradicting \eqref{eq:boundedness}.

[\textbf{(ii)} horizontal contraction]
We show that any pair of points in $\calC$ 
converge to the image of an integral curve of $v(x)$.
Consider any curve $\gamma_0(\cdot):[0,1] \to \calC$ connecting
two different points $\gamma_0(0), \gamma_0(1) \in \calC$.
The evolution of the curve along the flow 
is given by $\gamma_t(s) := \psi_t(\gamma(s))$.
Note that $\frac{d}{ds} \gamma_t(s) = \partial\psi_t(\gamma(s)) \dot{\gamma}(s)$.
From \textbf{(i)}, either 
$\lim\nolimits\limits_{t \to \infty} 
\frac{d}{ds} \gamma_t(s) = 0$ by \eqref{eq:local_convergence},
or there exists $T$ such that 
$
\frac{d}{ds} \gamma_t(s) \in -\calK(\gamma_t(s)) \cup\calK(\gamma_t(s))
$
for all $t\geq T$.
In this last case, using \eqref{eq:local_stability},
$\frac{d}{ds} \gamma_t(s)$ is bounded,
and 
$d_{\calK(\gamma_t(s))}\left( 
\frac{d}{ds} \gamma_t(s), 
\partial \psi_t(\gamma(s))  v(\gamma(s)) \right) = 0$.
Thus, in the limit of $t\to \infty$, 
$\frac{d}{ds} \gamma_t(s)$ becomes parallel to $v(\gamma_t(s))$.
It follows that 
$\psi_t(\gamma(s))$ belongs to the image of an integral curve of 
$v(x)$.

[\textbf{(iii)} uniqueness]
By contradiction,
let $\calA_1$ and $\calA_2$ be images of two 
distinct attractive integral curves of $v(x)$.
Take $x_1 \in \calA_1$, $x_2 \in \calA_2$,
and consider a new curve $\gamma$ connecting them.
Along the flow, $\psi_t(\gamma(s))$ 
converges to an integral curve of $v(x)$. 
Thus, by completeness, $\calA_1$ and $\calA_2$
are subsets of the image of the same 
integral curve of $v(x)$, contradicting the 
initial hypotesis.
\end{proofof}

\begin{proofof}\emph{Theorem \ref{thm:PB}.}
From the conditions of the theorem, there
exists $0 < c_1 < c_2$ such that 
$c_1 \leq |f(x)| \leq c_2 $ for all $x\in \calC$. 
Reasoning like in the proof of 
Theorem \ref{thm:one-dimensional-attractor}
by replacing $v(x)$ by $f(x)$,
we get
\begin{equation}
\label{eq:local_stability2}
\lim_{t\to\infty} |\partial \psi_t(x)\delta x|_{\psi_t(x)} < \infty
\end{equation}
for any $x\in \calC$ and $\delta x\in T_x\calX$.
Furthermore, we can conclude that there exists
a unique integral curve of the vector field $f(x)$
- a trajectory - whose image is an attractor
for every trajectory of $\Sigma$ from $\calC$. 
In what follows we show that this curve is closed,
by following the proof of \cite[Corollary 2]{Forni2014a_ver1}.

Let $\psi(\cdot,x)$ be such a curve, for some $x\in \calC$. 
$\psi(\cdot,x)$ does not converge to a fixed point
and belongs to a compact set, therefore there
is a point $x^* = \psi(t^*,x)$ whose neighborhood
$\calB_\rho(x^*)$ is visited by the trajectory
infinitely many times, for any $\rho >0$. 
Take $\rho>0$ sufficiently small and 
let $\calS\subseteq\calC$ be
a transversal manifold to the trajectories of $\Sigma$
in $\calB_\rho(x^*)$, namely
$f(x) \notin T_x \calS$ for all $x\in \calS\cap\calB_\rho(x^*)$
and $x^* \in \calS$.
Consider a sequence of time instants $t_k \to \infty$ for $k\to\infty$
such that $\psi(t_k,x^*) \in \calS\cap\calB_\rho(x^*)$,
and the subsequence of time instants 
$t_{k_j} \to \infty$ for $j\to\infty$
such that $\psi(t_{k_j},x^*) \in \calS\cap\calB_{\frac{\rho}{3}}(x^*)$,

Consider now any curve $\gamma(\cdot):[0,1]\to\calS\cap\calB_\rho(x^*)$
such that the length of 
$\ell(\gamma(\cdot)) = \int_0^1 |\dot{\gamma}(s)|_{\gamma(s)} ds$
is less than or equal to $\rho \in \real$ and
$\gamma(0) = x^*$.
By \eqref{eq:local_stability2},
there exists  $c_3\geq 1$ such that
$\limsup_{t\to\infty} \ell(\psi_t(\gamma(\cdot))) \leq c_3\rho$.
Furthermore, 
$
\lim_{t\to\infty}
d_{\calK(\psi_t(x))}
(\frac{d}{ds}\psi_t(\gamma(s)), f(\psi_t(\gamma(s))) = 0 
$
for all $s \in [0,1]$, that is, 
in the limit of $t\to \infty$, 
the image of $\psi_t(\gamma(\cdot))$ converges asymptotically to
the image of the attractive integral curve of the vector field $f(x)$.
Thus, because of the transversality of $\calS$,
the combination of the bound 
$\limsup_{t\to\infty} \ell(\psi_t(\gamma(\cdot))) \leq c_3\rho$
and of the convergence of the Hilbert metric 
guarantees  
that there exist
$k$ sufficiently large and $\tau_k(s)$ typically
small, such that
$\psi(t_k+ \tau_k(s),\gamma(s)) \in \calS\cap\calB_{\frac{\rho}{3}}(\psi(t_k,x^*))$,
$\forall s\in [0,1]$.
Thus, there exists a $j$ sufficiently large such that
$\psi(t_{k_j}+ \tau_{k_j}(s),\gamma(s)) \in 
\calS\cap\calB_{\frac{\rho}{3}}(\psi(t_{k_j},x^*))
\subset \calS\cap\calB_{\frac{2\rho}{3}}(x^*)$.
It follows that the return map on $\calS$
is necessarily a contraction, which implies
that the attractor is a closed curve.
\end{proofof}

\section{Extended example: all-to-all Kuramoto}
\label{sec:Kuramoto}

We study the 
synchronization of Kuramoto dynamics for the case of 
all-to-all coupling with homogeneous velocities at zero
using the results of Section \ref{sec:asymptotic_behavior}.
Differential positivity is verified by using
the geometric conditions of Section \ref{sec:geometric_conditions}.

Consider the interconnection of $n$ agents (phase variables)
given by 
$\dot{\vartheta}_k = 
 \frac{1}{n}\sum_{i=1}^n  \sin(\vartheta_i - \vartheta_k)
$ where $\vartheta_k \in \mathbb{S}$.
The synchronization manifold is given by 
$ \calA := \{\vartheta \in \mathbb{S}^n\,|\,
\vartheta_1 = \dots = \vartheta_n \} $,
where $\vartheta := \smallmat{\vartheta_1 & \dots & \vartheta_n}^T$,
$k \in \calI :=  \{1,\dots, n\}$.
Using $\mathbf{1} := \smallmat{1 & \dots & 1}^T$,
the prolonged dynamics reads
\begin{equation}
\left\{
\begin{array}{rcl}
 \dot{\vartheta} &=& 
 \frac{1}{n}\mathbf{S}(\vartheta)\mathbf{1} \\
 \dot{\delta \vartheta} &=& \frac{1}{n}\mathbf{C}(\vartheta) \delta \vartheta
\end{array}
\right.
\qquad (\vartheta, \delta \vartheta) \in T\mathbb{S}^n
\end{equation}
where for $k,i\in \calI $,
$\mathbf{S}_{ki}(\vartheta) := \sin(\vartheta_i-\vartheta_k)$,
$\mathbf{C}_{kk}(\vartheta) := -\sum_{i\neq k}\cos(\vartheta_i-\vartheta_k)$,
and  $\mathbf{C}_{ki}(\vartheta) := \cos(\vartheta_i-\vartheta_k)$
for $ k \neq j$.

The invariance of the synchronization manifold 
is captured by the identity $\mathbf{C}(\vartheta) \mathbf{1} = 0$.
The constant vector field $v(\vartheta) = \frac{\mathbf{1}}{\sqrt{n}}$ 
satisfies the conditions of 
Theorem \ref{thm:one-dimensional-attractor}. 
Thus, exploiting the theorem, 
given a compact and forward invariant 
region $\calC$,
all trajectories from $\calC$ converge 
asymptotically  to the synchronization manifold $\calA\subseteq \calC$,
provided that the system is strictly differentially positive
with respect to some cone field $\calK(\vartheta)$ such that
$\mathbf{1} \in \mathrm{int}\calK (\vartheta)$
(interior of $\calK(\vartheta)$).

Define $\Pi:= \left[I - \frac{\mathbf{1}\mathbf{1}^T}{n}\right]$
and consider the
centroid $\rho e^{i\phi} := \frac{1}{n}\sum_{k\in \cal I}e^{i\vartheta_k}$.
$\rho\in[0,1]$ is a measure of the synchrony of the phase variables
\cite[Section III]{Sepulchre07kuramoto}.
We show that the system is strictly differentially positive
in any compact and forward invariant region $\calC\subset \mathbb{S}^n$
with respect to the quadratic cone field 
\eqref{eq:quadratic_constraints_vectors} given by
\begin{equation}
\label{eq:quadratic_cone_kuramoto2}
\begin{array}{rcl}
K_1(\vartheta,\delta \vartheta) \!\!&\!\!:=\!\!&\!\!  \mathbf{1}^T \delta \vartheta \vspace{2mm} \\
K_2(\vartheta,\delta \vartheta) \!\!&\!\!:=\!\!&\!\!  
e^{2\lambda \rho}\delta \vartheta^T \mathbf{1} \mathbf{1}^T \delta \vartheta -
\delta \vartheta^T \Pi \delta \vartheta 
\end{array}
\end{equation}
where $\lambda > 0$ is a parameter defined below.
An illustration is provided in 
Figure \ref{fig:quadratic_kuramoto}.
\begin{figure}[htbp]
\centering
\includegraphics[width=0.25\columnwidth]{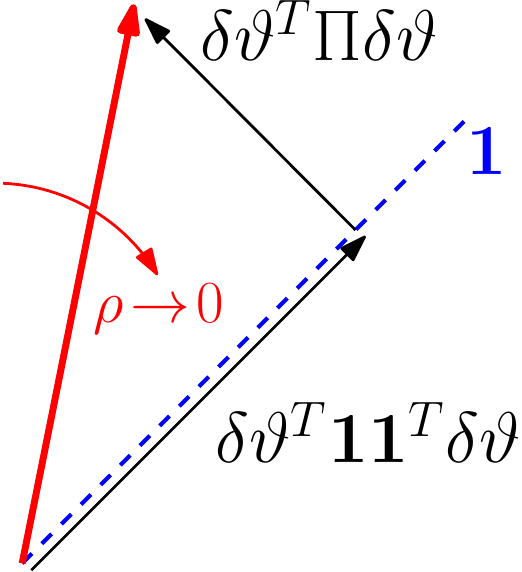}  
\caption{Increasing $\rho$ widens the cone. The cone is
widest at $\rho = 1$.}
\label{fig:quadratic_kuramoto}
\end{figure}

For $K_2(\vartheta,\delta \vartheta) = 0$, using
the identities $\mathbf{1}^T \mathbf{C}(\vartheta) = 0$
and $\Pi \mathbf{C}(\vartheta) = \mathbf{C}(\vartheta)$,
$\frac{d}{dt}K_2(\vartheta,\delta \vartheta)$ along the flow
of the prolonged dynamics reads
\begin{equation}
\label{eq:derivative_kuramoto}
2\lambda \dot{\rho} e^{2\lambda \rho}\delta \vartheta^T \mathbf{1} \mathbf{1}^T \delta \vartheta
-\delta \vartheta^T (\mathbf{C}(\vartheta)^T
+\mathbf{C}(\vartheta)) \delta \vartheta 
\end{equation}
where, following \cite[Section III]{Sepulchre07kuramoto} and \cite[Section VII]{Forni2014},
$
\dot{\rho} = \frac{\rho}{n} \sum_{k\in\calI} \sin(\vartheta_k-\phi)^2 
$.
In particular, $\dot{\rho}=0$ 
for $\rho =0$ (balanced phases, max spread on the circle) or for
$\sum_{k\in\calI} \sin(\vartheta_k-\phi)^2=0$, which occurs on isolated 
critical points given by $n-m$ phases synchronized at $\phi + 2j\pi$ and $m$ 
phases synchronized at $\phi+\pi+2j\pi$, for $j\in \mathbb{N}$ and $0\leq m\leq \frac{n}{2}$
(for an extended analysis see \cite[Section III]{Sepulchre07kuramoto}). 

For 
$\vartheta \in \mathbb{S}^n_{\pi/2} := 
 \left\{\vartheta \in \mathbb{S}^n\,|\, 
|\vartheta_k-\vartheta_j| < \frac{\pi}{2} \, , \ \forall k,j  \in \calI  \right\}$,
the quantity 
$-\delta \vartheta^T (\mathbf{C}(\vartheta)^T
+\mathbf{C}(\vartheta)) \delta \vartheta >0$. 
For $\vartheta \notin \mathbb{S}^n_{\pi/2}$,
we can design $\lambda$ so that
the first term 
$2\lambda \dot{\rho} e^{2\lambda \rho}\delta \vartheta^T \mathbf{1} \mathbf{1}^T \delta \vartheta$ in \eqref{eq:derivative_kuramoto}
dominates $-\delta \vartheta^T (\mathbf{C}(\vartheta)^T
+\mathbf{C}(\vartheta)) \delta \vartheta $.
Given any compact and forward invariant set $\mathcal{C}\subset \mathbb{S}^n$
that does not contain any balanced phase ($\rho = 0$) or saddle point
($\sum_{k\in\calI} \sin(\vartheta_k-\phi)^2=0$), 
there exists a sufficiently small
$\varepsilon >0$ such that 
$\sum_{k\in\calI} \sin(\vartheta_k-\phi)^2>\varepsilon$ and $\rho >0$ for every $\vartheta\in\mathcal{C}$. 
Thus, there exists a sufficiently large
$\lambda$ such that 
$K_2(\vartheta,\delta \vartheta) = 0 \Rightarrow \frac{d}{dt}K_2(\vartheta,\delta \vartheta) > 0 $. 

\section{Conclusion}
Differential positivity provides
a number of methods for
the analysis of simple attractors 
of nonlinear (closed) systems,
capturing bistable and periodic behaviors.
The introduction of compact sets 
simplifies the use of differential positivity in applications.
The theory have been illustrated on several examples 
based on polyhedral and quadratic cone fields.
Future research directions will 
study differential positivity on open systems. 

\bibliographystyle{plain}

\end{document}